%% file: Noise_SPL.tex
\documentclass[journal]{IEEEtran}
\usepackage{cite}

\usepackage[cmex10]{amsmath}
\usepackage{bm, amsthm}
\usepackage{amsfonts}
\usepackage{amssymb}
\usepackage[acronym,nomain]{glossaries}
\usepackage[final]{graphicx}
\usepackage{tablefootnote}
\usepackage{pgfplots}
\usepackage{filecontents}
\usepgfplotslibrary{fillbetween}
\pgfplotsset{compat=1.14}

\usepackage{algorithmic}

\usepackage{array}

\ifCLASSOPTIONcompsoc
 \usepackage[caption=false,font=normalsize,labelfont=sf,textfont=sf]{subfig}
\else
 \usepackage[caption=false,font=footnotesize]{subfig}
\fi

\newcommand{\changed}[1]{\textcolor{black}{#1}}

\usepackage{fixltx2e}
\usepackage{stfloats}

\newtheorem{lemma}{Lemma}

\newcolumntype{P}[1]{>{\centering\arraybackslash}p{#1}}
\newcolumntype{M}[1]{>{\centering\arraybackslash}m{#1}}
\newcolumntype{L}[1]{>{\raggedright\let\newline\\\arraybackslash\hspace{0pt}}m{#1}}
\newcolumntype{C}[1]{>{\centering\let\newline\\\arraybackslash\hspace{0pt}}m{#1}}
\newcolumntype{R}[1]{>{\raggedleft\let\newline\\\arraybackslash\hspace{0pt}}m{#1}}

  \pgfplotsset{ every non boxed x axis/.append style={x axis line style=-},
     every non boxed y axis/.append style={y axis line style=-}}

\include{commands}
\include{definitions}

\begin{document}
\title{On the SNR Variability in Noisy Compressed Sensing}


\author{\IEEEauthorblockN{Anastasia Lavrenko~\IEEEmembership{Student~Member,~IEEE},
Florian R\"omer~\IEEEmembership{Senior Member,~IEEE},\\
{Giovanni Del Galdo~\IEEEmembership{Member,~IEEE}},
and 
{Reiner Thom\"a~\IEEEmembership{Fellow,~IEEE}}}
\thanks{

 A. Lavrenko (corresponding author) and F. R\"omer  are with the Institute for Information Technology at Technische Universit\"at Ilmenau, Helmholzplatz 2, 98693, Ilmenau, Germany (email:\{anastasia.lavrenko, florian.roemer\}@tu-ilmenau.de). 
G. Del Galdo and R. Thom\"a are with the Institute for Information Technology at Technische Universit\"at Ilmenau and with the  Fraunhofer Institute of Integrated Circuits (IIS), Helmholzplatz 2, 98693, Ilmenau, Germany (email: giovanni.delgaldo@iis.fraunhofer.de, reiner.thomae@tu-ilmenau.de).
This work was partially supported by the Deutsche
Forschungsgemeinschaft (DFG) project CLASS (grant MA 1184/23-1) and the Carl-Zeiss Foundation under the postdoctoral scholarship project  "EMBiCoS".

}}


\IEEEtitleabstractindextext{%
\begin{abstract}
Compressed sensing (CS) is a sampling paradigm that allows to simultaneously measure and compress signals that are sparse or compressible in some domain.
The choice of a sensing matrix that carries out the measurement has a defining impact on the system performance and it is often advocated to draw its elements randomly. 
It has been noted that in the presence of input (signal) noise, the application of the sensing matrix causes SNR degradation due to the noise folding effect. In fact, it might also result in the variations of the output SNR in compressive measurements over the support of the input signal, potentially resulting in unexpected non-uniform system performance. In this work, we study the impact of a distribution from which the elements of a sensing matrix are drawn on the spread of the output SNR. We derive analytic expressions for several common types of sensing matrices and show that the SNR spread grows with the decrease of the number of measurements. This makes its negative effect especially pronounced for high compression rates that are often of interest in CS.
\end{abstract}

}

\maketitle

\IEEEdisplaynontitleabstractindextext

%
\IEEEpeerreviewmaketitle

\input{intro}


\input{mainbody}

\section*{Acknowledgment}
We would like to thank anonymous reviewers for their valuable comments and suggestions.

\section{Conclusions}

In this work, we have considered the influence of a choice of a sensing matrix on the output SNR in the noisy compressed sensing setting. We have demonstrated that for a fixed signal power, the application of the sensing matrix can potentially result in a varying output SNR that depends on the signal support. We have shown how the distribution from which the elements of the sensing matrix are drawn can be used to evaluate the spread of the output SNR on the example of several types of matrices widely used in CS. Our numerical results show good correspondence between the analytic and empirical analysis confirming the intuition that the SNR variations can be significant. \changed{Therefore, this effect should not be overlooked during the system design.} 





%

\newpage
\bibliographystyle{IEEEbib}
\bibliography{BibliogrP}

\end{document}

%% file: commands.tex
  \newcommand{\SuppM}[1]{{\mathcal{S}}_{#1}} 	
 \newcommand{\MTXT}[1]{{#1}^{\rm T}} 		
  \newcommand{\SensM}{\bm{A}}				
  \newcommand{\MeasM}{\bm{\Phi}}				
   \newcommand{\Basis}{\bm{\Psi}} 			
   
    \newcommand{\osnr}{\eta_{\rm O}}  
  \newcommand{\rsnr}{\eta_{\rm R}} 
  
  \newcommand{\meanof}[1]{\mathbb{E}\left\{{#1}\right\}}    
 \newcommand{\varof}[1]{\text{var}\left\{ {#1}\right\}}   

\definecolor{nicered}{RGB}{231,84,121}
\definecolor{darkblue}{RGB}{20,43,140}
\definecolor{darkred}{RGB}{216,41,0}
\definecolor{nicegreen}{RGB}{32,127,43}
\definecolor{darkbrown}{RGB}{78,50,8}

%% file: definitions.tex
\newacronym{cs}{CS}{Compressed Sensing}
\newacronym{smv}{SMV}{Single Measurement Vector} 
\newacronym{mmv}{MMV}{Multiple Measurement Vector}
\newacronym{mwc}{MWC}{Modulated Wideband Converter} 
\newacronym{emd}{EMD}{Earth Mover's Distance}
\newacronym{rmse}{RMSE}{Root Mean Squared Error}
\newacronym{rip}{RIP}{Restricted Isometry Property}
\newacronym{snr}{SNR}{Signal-to-Noise Ratio}
\newacronym{pdf}{PDF}{Probability Density Function}
\newacronym{cdf}{CDF}{Cumulative  Distribution Function}

\newglossaryentry{l-norm}
{
    name= $\|\bm{x}\|_p$,
    description={An $\ell_p$ norm of some vector $\bm{x}$}
}



%% file: intro.tex
\section{Introduction}

Recent developments in the areas of sampling theory 
and numerical optimization have recently given rise to the novel sampling framework of \gls{cs} \cite{donoho2006compressed, CandesEtAl2008, EldarEtAl2012CSbook}.  
Mathematically, its primary focus is solving the following under-determined system of linear equations  
\begin{equation}
	\bm{y} =\SensM \bm{x} + \bm{n},
	\label{eq:Ch2_MeasNoisy1}
\end{equation}
where $\bm{y} \in \mathbb{R}^{M \times 1}$ contains linear measurements of some signal $\bm{x} \in \mathbb{R}^{N \times 1}$ obtained via application of a sensing matrix $\bm{A} \in \mathbb{R}^{M \times N}$ with $M < N$, while $\bm{n} \in \mathbb{R}^{M \times 1}$ represents additive noise that is typically modeled either as deterministic and bounded \cite{CandesEtAl2006} or as white Gaussian \cite{Haupt2006NoisyReconstr, Ben-Haim2010Coherence}. Additionally, the input vector $\bm{x}$ in \eqref{eq:Ch2_MeasNoisy1} is assumed to be sparse, meaning that only $K \ll N$ of its elements are non-zero \cite{donoho2006compressed, CandesEtAl2008, EldarEtAl2012CSbook, DonohoEtAl2003}. A number of algorithms exists in the literature that allow to efficiently solve \eqref{eq:Ch2_MeasNoisy1} including greedy algorithms \cite{tropp2004greed} and methods based on convex relaxation \cite{tropp2006just}.


In the \gls{cs} setting, the choice of the sensing matrix has a defining impact on the reconstruction accuracy \cite{CandesEtAl2008, EldarEtAl2012CSbook}. 
It has been extensively studied especially with respect to recovery bounds in both noise-free and noisy settings.
Two particularly well-studied matrix fitness measures  are the matrix coherence \cite{elad2007optimized, CandesEtAl2008, Ben-Haim2010Coherence} and the \gls{rip} \cite{candes2008restricted, baraniuk2008simple, CandesEtAl2008}.   
Recently, several papers have also discussed the effect that the application of a sensing matrix has on the so-called input noise that is added to $\bm{x}$ prior to the measurement \cite{aeron2010information, ben2012performance}. Among these, the most prominent is the noise folding which shows itself in the increase of the  input noise power proportional to the compression ratio $\frac{N}{M}$ \cite{Arias-CastroEtAl2011, DavenportEtAl2012}. 

Another important effect arising from applying $\SensM$ in \eqref{eq:Ch2_MeasNoisy1} that has been largely overlooked so far is the variability of the output \gls{snr}. It turns out that the effective signal power in compressed measurements depends on the entries of the sensing matrix corresponding to the support of the input signal \cite{lavrenko2014sensing}. As a result, for a fixed input \gls{snr}, the output \gls{snr} becomes dependent on the signal support. 
\changed{
Bounding the recovered \gls{snr} for the best-case scenario of correct support recovery shows that this effect can potentially lead to support-dependent recovery guarantees. Subsequently, 
we can also expect a
non-uniform support recovery performance. 
%
To evaluate the extent of such \gls{snr} variations, we investigate the spread of the output \gls{snr} on the example of sensing matrices commonly used in \gls{cs}. To do so, we model the effective signal power as a random process whose characteristics are determined by the distribution from which the elements of $\SensM$ are drawn. 
We demonstrate that the coefficient of variation of the output \gls{snr} is inversely proportional to $\sqrt{M}$, which indicates that it can be significant for compression rates typical in \gls{cs}. Hence, this effect should be taken into account while designing the measurement, e.g., by choosing the sensing matrix that minimizes the SNR spread for a certain level of matrix coherence, or by choosing the number of measurements according to the lower  (worst-case) attainable SNR levels for a given $\SensM$.}  


%% file: mainbody.tex
\section{Noisy Compressive Sensing}
\label{sec:NCS}


\subsection{Noise Model}
\label{sec:SysM}

In the noisy \gls{cs} setting, two types of additive noise occur:  signal or input noise that represents noise sources acting before the compression takes place, i.e., before the application of $\SensM$, and measurement noise that accounts for the noise sources that act afterwards \cite{Arias-CastroEtAl2011, DavenportEtAl2012}. 
In light of that, we write \eqref{eq:Ch2_MeasNoisy1} as
\begin{equation}
	\bm{y} =\SensM \bm{x} + \bm{n} 
	= \SensM( \bm{x} + \bm{n}_{\rm s}) + \bm{n}_{\rm m},
	\label{eq:Ch2_MeasNoisy2}
\end{equation}
where
$	\bm{n} = \SensM \bm{n}_{\rm s}+ \bm{n}_{\rm m}$,
while $\bm{n}_{\rm s}$ and $\bm{n}_{\rm m}$ denote the signal and the measurement noise, respectively. 
%
%
We assume throughout that $\bm{n}_{\rm s}$ and $\bm{n}_{\rm m}$ are independent random vectors with i.i.d.~zero-mean Gaussian (normal) distributed elements with variance $\sigma_{\rm s}^2$ and $\sigma_{\rm m}^2$, respectively.

	Under 
	this assumption, the covariance matrix $\bm{\Sigma}$ of the total noise vector $\bm{n}$ becomes
\begin{equation}
	\bm{\Sigma} = \sigma_{\rm s}^2 \SensM \MTXT{\SensM} + \sigma_{\rm m}^2{\mathbf{I}}_{M}, 
	\label{eq:NoiseCovG}
\end{equation}
where ${\mathbf{I}}_{M}$ represents an $M \times M$ identity matrix and $(\cdot)^{\rm T}$ denotes the matrix transpose. Expression \eqref{eq:NoiseCovG} shows a first consequence of the application of $\SensM$ to $\bm{x}$: the coloring of the signal noise $\bm{n}_{\rm s}$  when $\SensM \MTXT{\SensM} \neq c {\mathbf{I}}_{M}$. 
%
In a special case, when the rows of $\SensM$ are orthogonal with an equal norm of $ \sqrt{\frac{N}{M}}$, the noise $\bm{n}$ is white with covariance
\begin{equation}
	\bm{\Sigma} = \frac{1}{\rho}\sigma_{\rm s}^2 {\mathbf{I}}_{M} + \sigma_{\rm m}^2{\mathbf{I}}_{M} = \left(\frac{1}{\rho}\sigma_{\rm s}^2 + \sigma_{\rm m}^2\right){\mathbf{I}}_{M}, 
	\label{eq:NoiseCovWhite}
\end{equation}
where $\rho = \frac{M}{N}$.
From \eqref{eq:NoiseCovWhite}, it can be seen that the variance of the signal noise after compression increases by the factor of $\frac{1}{\rho}= \frac{N}{M}$. This is because the sensing matrix  $\SensM$ combines the input noise along the entire $N$-dimensional space, whereas the signal resides in its $K$-dimensional sub-space. The resulting increase of the signal noise power  in the compressed measurements is known as the \textit{noise folding} effect \cite{Arias-CastroEtAl2011, DavenportEtAl2012}.

\vspace*{-0.2cm}
\subsection{\gls{snr} Measures}
 To this end, we 
 denote by $\SuppM{\bm{x}}$ the support of $\bm{x}$ and by \changed{$P_{\rm s} = \meanof{\|\bm{x}\|_2^2}$} the total signal power.
\changed{Given 
\eqref{eq:Ch2_MeasNoisy2}, several types of \gls{snr} can be considered \cite{DavenportEtAl2012}.
%
Among these, the so-called output \gls{snr} is of special importance as it expresses}
 the ratio between the total signal power after the measurement to the total noise power:
\begin{align}
\label{eq:SNRN}
	\eta_{{\rm O}} \overset{\Delta}{=}   \frac{\| \SensM \bm{x} \|^2 _2}{\mathbb{E}\{\|\bm{n}\|^2_2\}} &= \frac{\| \SensM \bm{x} \|^2 _2}{\mathbb{E}\{\|\SensM \bm{n}_{\rm s} + \bm{n}_{\rm m}\|^2_2\}} \nonumber \\
            &  =\frac{ 1}{M\sigma_0^2}\sum_{m=1}^{M} \left(\sum_{i \in
            \SuppM{\bm{x}}}{a}_{m,i}  {x_i} \right)^2 ,
\end{align}
where 
$	\sigma_0^2 =  \frac{1}{M}\text{trace}\{ \bm{A}\bm{A}^{\rm T}\} \sigma^2_{\rm s}+\sigma_{\rm m}^2$.

Expression \eqref{eq:SNRN} reveals another important effect arising from applying the sensing matrix $\SensM$ in \eqref{eq:Ch2_MeasNoisy2}, namely the  dependency of the effective signal power on the entries of the sensing matrix corresponding to the support of the input signal. 
Note that in most applications, the sensing matrix would be fixed at least for some time after its elements are chosen (e.g., drawn according to some probability distribution), since a truly random measurement is often impractical from the hardware viewpoint. 
  \changed{Therefore, for fixed noise powers,  $\osnr$ is generally a function of two variables: the support $\SuppM{\bm{x}}$ via the corresponding values $a_{m,i}$ and the non-zeros $x_i$  that we arrange into a sequence  $\mathcal{X} = \{x_{i_1},\dots, x_{i_K}\}$ where $\forall k\in [1, K-1] \,  i_k< i_{k+1} \in \SuppM{\bm{x}}$. 
The fact that the magnitudes of $\bm{x}$ have an impact on the SNR is not surprising  as the SNR is meant to be a measure of the signal power with respect to the noise. 	
What distinguishes \eqref{eq:SNRN}, is that, other things being equal, the change of the signal support  can lead to the change of the output\footnote{\changed{This does not occur in the traditional Nyquist-rate sensing when $\SensM = \mathbf{I}_N$.}} SNR.}
As a result, the effective SNR might vary depending on the positions of the non-zeros in $\bm{x}$ leading to  potentially non-uniform (over the support of the input signal) system performance. 

\changed{
To illustrate this, consider another \gls{snr} measure known as the recovered \gls{snr} \cite{DavenportEtAl2012}. Defined as 
\begin{equation}
	\eta_{{\rm R}} \overset{\Delta}{=} \frac{ \| \bm{x} \|^2_2}{ \mathbb{E}\{\|\hat{\bm{x}}-\bm{x}\|^2_2\}},
	\label{eq:rsnr_def}
\end{equation}
it accounts for the ratio of the signal power to the power of the residual noise present after reconstruction.
Naturally, the recovered \gls{snr} largely depends on the particular algorithm used to solve \eqref{eq:Ch2_MeasNoisy1}. To circumvent this, we adopt an oracle-assisted approach to performance evaluation  that assumes that the support of $\bm{x}$ is known prior to the recovery \cite{haupt2009compressive, DavenportEtAl2012, laska2012regime}. In doing so, we evaluate the best-case performance which sets a benchmark for any practical recovery method.}

\changed{
Once the true support $\SuppM{\bm{x}}$ is known, we can write \eqref{eq:rsnr_def} as
\begin{equation}
            \rsnr = \frac{\|\bm{x}\|_2^2}{\meanof{\| \SensM^{\dagger}_{\SuppM{\bm{x}}} \bm{y} - \bm{x}  \|_2^2}} = \frac{\|\bm{x}\|_2^2}{\meanof{\| \SensM^{\dagger}_{\SuppM{\bm{x}}} \bm{n} \|_2^2}},
            \end{equation}}
%
\changed{
The ratio of $\rsnr$ to $\osnr$ can then be  bounded \cite{DavenportEtAl2012} as
\begin{equation}
    \left(\frac{1-\delta}{1+\delta} \right) \frac{M}{K}     \leq    \frac{\rsnr}{\osnr} \leq \left(\frac{1+\delta}{1-\delta}\right) \frac{M}{K},
    \label{eq:rsnr_bound_oracle}
\end{equation}
where $\delta \in (0,1)$ is the \gls{rip} constant \cite{candes2008restricted}.
}
\changed{
Inequality \eqref{eq:rsnr_bound_oracle} shows that bounds on the best-case recovered \gls{snr}  scale linearly with the output \gls{snr}. This in turns indicates that the variation of the output \gls{snr} with respect to the signal support will result in a corresponding variation of the bounds on $\rsnr$; hence  we can expect a non-uniform (best-case) recovery performance over different signal supports. Furthermore, as the support estimation is the most challenging aspect of sparse recovery it is reasonable to suspect that in practice such an \gls{snr} spread might have an even more dramatic impact. }

\changed{
The goal of this study is to investigate this particular effect,  the variation of the output SNR over the support of $\bm{x}$. We are particularly interested in the impact of the choice of $\SensM$ on the spread of $\osnr$ as this is  what differentiates the CS approach form the traditional one.} 



\section{Analysis of the Output \gls{snr}}
\label{sec:SNRAn}

\subsection{SNR Spread Evaluation}

From \eqref{eq:SNRN}, the output SNR $\osnr$ depends on the support of $\bm{x}$ via $\beta =\sum_{m = 1}^M \left(\sum_{k \in \SuppM{\bm{x}}}{a}_{m,k}{x_k}\right)^2$. 
\changed{
When $\SensM$ is fixed, all $a_{m,n}$ are deterministic. Hence, for any  given $\mathcal{X}$, we can potentially compute a conditional frequency distribution $h_{\SensM}(\osnr | \mathcal{X})$, as well as the sample mean and the sample variance of $(\osnr| \mathcal{X})$, by going through all possible supports $\SuppM{\bm{x}}$. The notation $h_{\SensM}(z)$ here indicates that the distribution of $z$ is subject to change with the change of $\SensM$.
To account for different $\mathcal{X}$, we can repeat this procedure for different combinations of signal magnitudes and average the results. This would result in a marginal frequency distribution $h_{\bm{A}}(\osnr)$ over the support $\SuppM{\bm{x}}$.
}
Although evaluating the spread of the output SNR  this way allows us to characterize a particular realization of $\SensM$, it might become computationally unfeasible \changed{as calculating $h_{\SensM}(\osnr | \mathcal{X})$ even for a single choice of magnitudes already requires checking $C_N^K=\frac{N!}{K! (N-K)!}$ possible combinations which is known to be NP-hard.}
When the elements of $\SensM$ are drawn from some probability distribution and $N$ is large enough, we can eliminate this difficulty by approximating 
\changed{the frequency distribution $h_{\SensM}(\osnr | \mathcal{X})$} by the (analytic) probability distribution \changed{$f_{\SensM}(\osnr | \mathcal{X})$ derived by modelling the elements of $\SensM$ as~i.i.d. random variables.} 
In the following, we do so on the example of sensing matrices common\footnote{
Practically, one often measures not the sparse signal $\bm{x}$ itself but its representation in some basis $\Basis$. In this case, we arrive at the canonical \gls{cs} model of \eqref{eq:Ch2_MeasNoisy1} by expressing $\SensM$ as  $\SensM = \MeasM^{\rm T} \Basis$ where $\MeasM$ is now a matrix to be designed. This way, when the elements of $\MeasM$ are drawn according to some distribution and $\Basis$ is known, we can determine the distribution of $\bm{A}$ 
by analyzing the corresponding random variables $a_{m,k} = \bm{\phi}^{\rm T}_m \bm{\psi}_k$ where $\bm{\phi}_i$  and $\bm{\psi}_i$ denote $i$th columns of $\MeasM$ and $\Basis$, respectively. \changed{Note that in this case, another common choice of $\MeasM$ is the (random) selection matrix.}} in \gls{cs}, namely Gaussian, Bernoulli and Rademacher $\SensM$.

\subsection{Analytic Analysis}

\changed{
Note that when $a_{m,n}$ are independently drawn from some probability distribution $f(\alpha)$,  each row of $\SensM$ can be interpreted as containing an $N$-point sample from $f(\alpha)$}. 
Writing \eqref{eq:SNRN} as
\begin{equation}
	 \eta_{\rm O} =\frac{ 1}{M\sigma_0^2} \beta 
	 = \frac{ 1}{M\sigma_0^2}\sum_{m=1}^{M} d_{m}^2,
	 \label{eq:OSNRK}
\end{equation}
where $d_{m} = \sum_{k \in \SuppM{\bm{x}}}{a}_{m,k}  {x_k}$, \changed{we see that for a given $\mathcal{X}$, $d_m$ is a linear combination of $K$ realizations $a_{m, k}$ of some random variable $\alpha_m \sim f(\alpha)$. The support $\SuppM{\bm{x}}$ in this case defines which specific $K$ out of $N$ subset of realizations is taken. This enables the approximation of the frequency distribution $h_{\SensM}(\osnr | \mathcal{X})$ calculated for $M$ particular sets of realizations $\{a_{m,n}\}_{n=1}^N$ by  the probability distribution $f_{\SensM}(\osnr | \mathcal{X})$ computed under the assumption that $\{\alpha_m\}_{m=1}^M$ are~i.i.d.~random variables\footnote{\changed{We can do so for any deterministic sparsity pattern model including structured models such as block sparsity for instance. The analytic distribution $f_{\SensM}(\osnr| \mathcal{X})$ will not change in this case, whereas the approximation quality will deteriorate with the decrease in the size of the set of possible supports. On the other hand, imposing some probabilistic constraints on $\SuppM{\bm{x}}$ will require considering a joint distribution $f(\SensM, \SuppM{\bm{x}})$ to derive $f_{\SensM, \SuppM{\bm{x}}}(\osnr| \mathcal{X})$.}}. Once $f_{\SensM}(\osnr | \mathcal{X})$ is known we can marginalize $\mathcal{X}$ out to obtain 
\begin{equation}
     f_{\SensM}(\osnr) = \int_{\mathcal{X}} f_{\SensM}(\osnr|\mathcal{X}) f(\mathcal{X}) \mathsf{d}\mathcal{X} = \mathbb{E}_{\mathcal{X}}\{f_{\SensM}(\osnr|\mathcal{X})\},
     \label{eq:cond}
\end{equation}
where $\mathbb{E}_{\mathcal{X}}\{ \cdot \}$ means the average over the ensemble of  $\mathcal{X}$. Moreover, to evaluate the spread of $\osnr$ it is sufficient to calculate the mean $\mathbb{E}\{\osnr\} = \mathbb{E}_{\mathcal{X}}\{\mathbb{E}_{\SensM}\{\osnr| \mathcal{X}\}\}$ and the variance ${\rm var}\{\osnr\} = \mathbb{E}\{\osnr^2\} - \mathbb{E}^2\{\osnr\}$.}

\subsubsection{Gaussian $\SensM$}
Suppose the elements of $\SensM$ are drawn from a zero-mean normal Gaussian distribution such that $a_{m,n} \sim \mathcal{N}(0, 1/{M})$. Then, $d_{\rm m}$ is a zero-mean normal variable with variance $\sum_{k=1}^K {x_k}^2/M = \|\bm{x}\|_2^2/M$, whereas $\beta \frac{M}{\|\bm{x}\|_2^2}$ is a random variable distributed according to the chi-squared distribution with $M$ degrees of freedom, i.e., $\beta \frac{M}{\|\bm{x}\|_2^2} \sim \chi_M^2$. 
Consider now the following Lemma.
\begin{lemma}
	Denote by  $\Gamma(k, \theta)$ a Gamma distribution with a shape parameter $k$ and a scale parameter $\theta$. If $Y \sim \chi^2_L$  and $c$ is a positive constant, then $cY \sim \Gamma \left( \frac{L}{2}, \frac{2}{c} \right)$. 
	\label{thm:LemmaSNR1}
\end{lemma}
\begin{proof}
See \cite{taboga2012lectures}.	
\end{proof}

From Lemma \ref{thm:LemmaSNR1}, we have that
$\beta \sim \Gamma\left(\frac{M}{2}, \frac{2 \|\bm{x}\|_2^2}{M}\right)$ and hence
\begin{equation}
	\changed{f_{\SensM}(\eta_{\rm O}| \mathcal{X}) =} \Gamma \left( \frac{M}{2},  \frac{2 \changed{\|\bm{x}\|_2^2}}{M^2\sigma_0^2} \right). 
 \label{eq:OSNR_distr}
\end{equation}
 The mean and variance of $\changed{(}\eta_{\rm O}\changed{|\mathcal{X})}$ are given by
	$\mathbb{E}\changed{_{\SensM}}\{\osnr\changed{| \mathcal{X}}\} = k\theta = \frac{\changed{\|\bm{x}\|_2^2}}{M\sigma_0^2}$
and
${\rm var}\changed{_{\SensM}}\{\osnr\changed{| \mathcal{X}}\} = k\theta^2 = \frac{2}{M} \changed{\mathbb{E}^2_{\SensM}\{\osnr| \mathcal{X}\}}$,
respectively. 
\changed{From \eqref{eq:OSNR_distr}, $f_{\SensM}(\osnr | \mathcal{X})$ depends on $\mathcal{X}$ only through $\|\bm{x}\|_2^2$. Therefore,}
\begin{equation}
\changed{
     \meanof{\osnr} = \mathbb{E}_{\mathcal{X}}\{ \mathbb{E}_{\SensM}\{ \osnr | \mathcal{X}\}\} = \mathbb{E}_{\mathcal{X}}\left\{ \frac{\|\bm{x}\|_2^2}{M\sigma_0^2} \right\} = \vartheta,}
     \label{eq:mean_G}
\end{equation}
where $\vartheta = \frac{P_{\rm s}}{M\sigma_0^2}$ denotes the ratio of the total signal power to the total system noise power. 
\changed{
Finally, we note that due to \eqref{eq:OSNR_distr}, $(\osnr^2 | \mathcal{X})$ is distributed according to  the generalized Gamma distribution with parameters $p = 0.5$, $d = 0.5k$ and $a = \theta^2$. Taking this into account, we have that $\mathbb{E}_{\SensM}\{\osnr^2| \mathcal{X}\} = \theta^2 \frac{\Gamma(k+2)}{\Gamma(k)} = \theta^2 (k+1)k$ and hence
\begin{equation}
    \meanof{\osnr^2} 
    = \left(\frac{M}{2}+1\right)\frac{2}{M} \vartheta^2 = \left(1+\frac{2}{M}\right)\vartheta^2.
\end{equation}
Therefore, the variance of $\osnr$ can be calculated as
\begin{align}
      \varof{\osnr} 
      = \left(1+\frac{2}{M}\right)\vartheta^2 - \vartheta^2 
      = \frac{2}{M} \vartheta^2.
      \label{eq:var_G}
\end{align}
}
\changed{Given \eqref{eq:mean_G} and \eqref{eq:var_G}, we can calculate the coefficient of variation for a Gaussian $\SensM$  as
\begin{equation}
    c_{{\rm v}}(\osnr) = \sqrt{\frac{\varof{\osnr}}{(\meanof{\osnr})^2}} = \sqrt{\frac{2}{M}}.
    \label{eq:CoefVarGaus}
\end{equation}
Note that the \gls{snr} spread in this case depends only on the number of measurements $M$. 
}

\subsubsection{Bernoulli $\SensM$}
Consider now a  Bernoulli distributed $\SensM$, i.e.,  $a_{m,n} \sim \mathcal{B}_e(p)$. Each $a_{m, k} x_k$ in this case is distributed as scaled Bernoulli where the scaling depends on the signal values. 
%
The distribution of the $K$-term sum of such random variables 
can be described directly via  probabilities as 
\begin{equation}
    \text{Pr}\left(d_m = \sum_{k=1}^K b_{n,k}  x_k \right) = p^n(1-p)^{K-n},
\label{eq:pr_bm_bernouli}
\end{equation}
where $b_{n,k} \in \{0, 1\}$ is a $k$th element of a binary vector $\bm{b}_n$ of length $K$ that contains exactly $n \in [0, K]$ ones
. From \eqref{eq:pr_bm_bernouli}, the distribution of $\beta = \sum_{m =1}^M d_m^2 $ depends on the choice of $\mathcal{X}$  
and, generally, does not converge to any well-defined form. However, in the special case when all non-zero elements of $\bm{x}$ are equal, 
$\changed{(}d_m\changed{|\mathcal{X})}$ becomes scaled Binomial distributed 
which allows us to calculate the mean and variance of $d_m^2$ as
\begin{align}
\label{eq:mean_dm_bern}
	\mathbb{E}\{d_m^2\} &= \text{var}\{d_m\}+\mathbb{E}\{d_m\}^2 = \big((K-1)p+1\big)pP_{\rm s} \\
	\text{var}\{ d_m^2\} &= 
	  \frac{1+ 2p\big(K-1\big)\big((2K-3)p+3\big)}{K}(1-p)pP_{\rm s}^2.
	 \label{eq:var_dm_bern}
\end{align}
\changed
{Since $d_m^2$ are independent, we have that $\meanof{\beta} = M \meanof{d_m^2}$ and $\varof{\beta} = M\varof{d_m^2}$. Therefore, we obtain
\begin{align}
\notag
     c_{{\rm v}}(\osnr) &= \sqrt{\frac{\varof{d_m^2}}{M(\meanof{d_m^2})^2}} \\
     &=\sqrt{\frac{1}{M} \frac{1+ 2p\big(K-1\big)\big((2K-3)p+3\big)}{K\big((K-1)p+1\big)^2p}(1-p)}.
     \label{eq:cv_bern}
\end{align}
Comparing \eqref{eq:CoefVarGaus} with \eqref{eq:cv_bern}, we can see that, additionally to $M$, the SNR spread for a Bernoulli sensing matrix also depends on the number of signals $K$.
}

\subsubsection{Rademacher $\SensM$}
A similar situation occurs with the Rademacher distributed $\SensM$. Since each $a_{m, k} x_k$ can take the value of $x_k$ or $-x_k$ with equal probability, the distribution of $d_m$ and, therefore that of $\beta$ and $\eta_{\rm O}$, inevitably depends on the signal values when $K>1$. For the case of equal signal magnitudes, one can show however that the mean and variance of $d_m^2$ becomes $KP_{\rm s}$ and $2K(K-1)P_{\rm s}^2$, respectively. \changed{This results in the following coefficient of variation 
\begin{equation}
     c_{\rm v}(\osnr) = \sqrt{\frac{2MK(K-1) P_{\rm s}^2}{M^2 K^2 P_{\rm s}^2}} = \sqrt{\frac{1}{M}\frac{2(K-1)}{K}}.
\end{equation}}

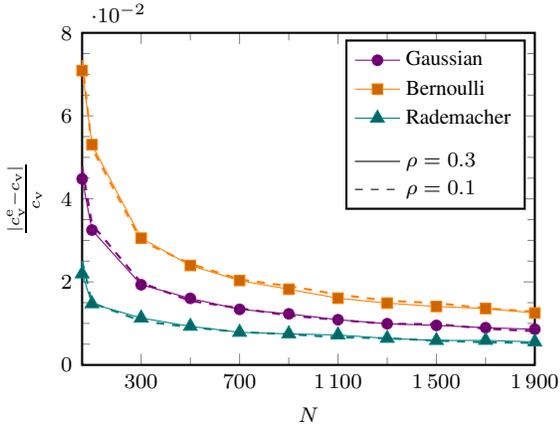
\begin{figure}[t]
\vspace*{-\baselineskip}
\vspace*{-0.05cm}
\centering
 \input{CoeFvar_MSE_fT.tex}
\vspace*{-0.3cm}
 \caption{\changed{Normalized RMSE between the empirical and the analytic coefficient of variation for $K=2$.}}\medskip
 \label{fig:CoefVar_MSE}
  \vspace*{-0.4cm}
\end{figure}

\section{Numerical Evaluation}
In this section, we numerically demonstrate the influence of the choice of the sensing matrix on the output \gls{snr}. To do so, we first generate $\bm{A}$ such that its entries are drawn from one of the  considered probability distributions. Then, for each realization of $\bm{A}$ we compute the output \gls{snr} $\eta_{\rm O}$  according to \eqref{eq:SNRN} for different supports $\SuppM{\bm{x}}$ \changed{and different choices of $\mathcal{X}$}. 

\changed{
To validate the derived analytic expressions, we calculate the average (among the realizations of the sensing matrix) \gls{rmse} between the coefficients of variation obtained numerically ($c_{\rm v}^{\rm e}$) and analytically ($c_{\rm v}$) for the case of equal signal magnitudes. Figure~\ref{fig:CoefVar_MSE} displays normalized \gls{rmse} as a function of $N$ 
for $K = 2$ and two values of $\rho$. In all cases, the \gls{rmse} does not exceed $7\%$ of the predicted value with the best correspondence exhibited by the Rademacher $\SensM$. As expected, the error decays as $N$ grows due to the increasing accuracy of the  analytic approximation, while being independent of the compression rate $1/\rho$.} This is due to the fact that the higher the dimension $N$ is, the more representative each row of $\bm{A}$ is  of $f(\alpha)$.

\changed{
To investigate how the choice of $\mathcal{X}$ influences the coefficient of variation, Figure~\ref{fig:CoefVar_EMPR} shows average (among $10^3$ realizations of $\bm{A}$ and $\mathcal{X}$) empirical coefficient of variation $c_{\rm v}^{\rm e}$ normalized to $1/\sqrt{M}$ as a function of $K$ computed for three considered types of $\SensM$ and different models on $\mathcal{X}$, namely i) all $x_i$ have equal magnitudes; ii) $x_i$ are ~i.i.d.~random variables distributed according to $\mathcal{N}(0, 1/{K})$; iii) $x_i$ are ~i.i.d.~random variables distributed according to $\mathcal{U}_{[-\sqrt{3/K}, \sqrt{3/K}]}$. 
As expected, the coefficient of variation for Gaussian $\SensM$ depends on neither the value of $K$ nor the type of $\mathcal{X}$ and it is equal to $c_{\rm v}^{\rm e} \cdot \sqrt{M} = \sqrt{2}\approx 1.4$. As for the Bernoulli and Rademacher sensing matrices, the results for equal magnitudes somewhat differ from those for other choices of $\mathcal{X}$ reflecting the dependency of $f_{\SensM}(\osnr|\mathcal{X})$ on $\mathcal{X}$.  Nevertheless,  they provide a lower spread of the output SNR than that of the Gaussian $\SensM$. }
It is worth nothing that in all considered cases the coefficient of variation is inversely proportional to $M$  leading to higher values for compression rates typical in \gls{cs}.


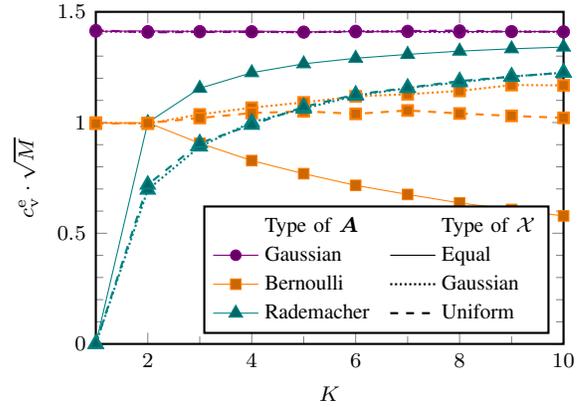
\begin{figure}[t]
\vspace*{-0.3cm}
\centering
 \input{CoeVar_EMPR.tex}
\vspace*{-0.3cm}
 \caption{\changed{ Average $c_{\rm v}^{\rm e}$ for equal (solid lines), Gaussian (dotted lines), and  Uniform (dashed lines) magnitudes with $N =300$.}}\medskip
 \label{fig:CoefVar_EMPR}
  \vspace*{-0.4cm}
\end{figure}

%% file: CoeFvar_MSE_fT.tex
\begin{center}
 
\begin{tikzpicture}[baseline]
\begin{axis}[height = 6cm, width =7.6 cm,, enlargelimits=false, 
minor x tick num=1,
minor y tick num=3,
ymax = 0.08,
ymin = 0,
xtick={300, 700, 1100, 1500, 1900},
legend cell align = left,
x tick label style={/pgf/number format/.cd,%
          scaled x ticks = false,
          set thousands separator={\,},
          fixed},
ylabel near ticks, xlabel near ticks, font=\footnotesize, thick, xlabel style={align=center}, xlabel={$N$
}, ylabel={$\frac{|c_{\rm v}^{\rm e} - c_{\rm v} |}{c_{\rm v}}$}]


\addplot[thin, violet, mark=otimes*,mark options = {fill=violet!80!black}] table[x index=0,y index=1, col sep=space] {CoeFvar_MSE_Gaus_03rho2K.txt};
\addlegendentry{Gaussian}

\addplot[thin, color = orange, mark =square*, mark options = {fill=orange!80!black}]
table[x index=0,y index=1, col sep=space] {CoeFvar_MSE_Bern_03rho2K.txt};
\addlegendentry{Bernoulli}

\addplot[thin, teal, mark =triangle*, mark size = 3 pt, mark options = {fill=teal!80!black}] table[x index=0,y index=1, col sep=space] {CoeFvar_MSE_Rad_03rho2K.txt};
\addlegendentry{Rademacher}

\addplot[draw=none, thin, color = white] coordinates {(1,1)};
\addlegendentry{}

\addplot[draw=none, thin, color = black] coordinates {(1,1)};
\addlegendentry{$\rho = 0.3$}
\addplot[draw=none, thin, dashed, color = black] coordinates {(1,1)};
\addlegendentry{$\rho = 0.1$}

\addplot[dashed, violet] table[x index=0,y index=1, col sep=space] {CoeFvar_MSE_Gaus_rho2K.txt};
\addplot[color = orange, dashed]
table[x index=0,y index=1, col sep=space]{CoeFvar_MSE_Bern_rho2K.txt};
\addplot[dashed, teal] table[x index=0,y index=1, col sep=space] {CoeFvar_MSE_Rad_rho2K.txt};

\end{axis}
\end{tikzpicture}

\end{center}

%% file: CoeVar_EMPR.tex
\begin{center}
\begin{tikzpicture}
\begin{axis}[height = 6cm, width =7.8 cm, 
minor x tick num=1,
minor y tick num=4,
legend pos=south east,
legend cell align = left,
legend columns=2,
ymax = 1.5,
ymin = 0,
xmax = 10,
xmin = 1,
ylabel near ticks, xlabel near ticks, font=\footnotesize, thick, xlabel style={align=center}, xlabel={$K$}, ylabel={$c_{\rm v}^{\rm e} \cdot \sqrt{M}$},
tick label style={/pgf/number format/fixed}]

\addplot[draw=none, color = white] coordinates {(1,1)};
\addlegendentry{Type of $\SensM$}
\addplot[draw=none, color = white] coordinates {(1,1)};
\addlegendentry{Type of $\mathcal{X}$}

\addplot[thin, color = violet, mark =*, mark options = {fill=violet!80!black}]
table[x index=0,y index=1, col sep=space] {CoeFvar_Gaus_rho1K.txt};
\addlegendentry{Gaussian}
\addplot[draw=none, thin, color = black] coordinates {(1,1)};
\addlegendentry{Equal}

\addplot[thin, color = orange, mark =square*, mark options = {fill=orange!80!black}]
table[x index=0,y index=1, col sep=space] {CoeFvar_Bern_rho1K.txt};
\addlegendentry{Bernoulli}

\addplot[draw=none, densely dotted, color = black] coordinates {(1,1)};
\addlegendentry{Gaussian}

\addplot[thin, teal, mark =triangle*, mark size = 3pt, mark options = {fill=teal!80!black}] table[x index=0,y index=1, col sep=space] {CoeFvar_Rad_rho1K.txt};
\addlegendentry{Rademacher $\;$}

\addplot[draw=none, dashed, color = black] coordinates {(1,1)};
\addlegendentry{Uniform}

\addplot[color = violet, densely dotted, mark =*, mark options = {solid, fill=violet!80!black}]
table[x index=0,y index=1, col sep=space] {CoeFvar_Gaus_rho1K_g.txt};
\addplot[color = violet, dashed, mark =*, mark options = {solid, fill=violet!80!black}]
table[x index=0,y index=1, col sep=space] {CoeFvar_Gaus_rho1K_u.txt};

\addplot[color = orange, densely dotted, mark =square*, mark options = {solid, fill=orange!80!black}]
table[x index=0,y index=1, col sep=space] {CoeFvar_Bern_rho1K_g.txt};
\addplot[color = orange, dashed, mark =square*, mark options = {solid, fill=orange!80!black}]
table[x index=0,y index=1, col sep=space] {CoeFvar_Bern_rho1K_u.txt};

\addplot[ densely dotted, teal, mark =triangle*, mark options = {solid, mark size = 3pt, fill=teal!80!black}] table[x index=0,y index=1, col sep=space] {CoeFvar_Rad_rho1K_g.txt};
\addplot[ dashed, teal, mark =triangle*, mark options = {solid, mark size = 3pt, fill=teal!80!black}] table[x index=0,y index=1, col sep=space] {CoeFvar_Rad_rho1K_u.txt};




\end{axis}
\end{tikzpicture}
\end{center}